\documentclass[conference]{IEEEtran}
\IEEEoverridecommandlockouts

\usepackage{cite}
\usepackage{amsmath,amssymb,amsfonts}
\usepackage{algorithmic}
\usepackage{graphicx}
\usepackage{textcomp}
\usepackage{xcolor}
\usepackage{amsmath, amssymb, amsthm}
\usepackage{url}

\newtheorem{theorem}{Theorem}
\newtheorem{lemma}[theorem]{Lemma}
\newtheorem{definition}[theorem]{Definition}

\newtheorem{remark}[theorem]{Remark}
\def\BibTeX{{\rm B\kern-.05em{\sc i\kern-.025em b}\kern-.08em
    T\kern-.1667em\lower.7ex\hbox{E}\kern-.125emX}}
\begin{document}

\title{Polarization on Bounded Euclidean Spaces\\ and More General Compact Groups\\
\thanks{L.G. and I.K. were supported in part by the EPSRC funded INFORMED-AI project EP/Y028732/1. S.M.S.\ and A.B.W.\ were supported in part by the US NSF under grant CCF-2306278.}
}

\author{%
\IEEEauthorblockN{Sharang M. Sriramu}
\IEEEauthorblockA{Fraunhofer Heinrich Hertz Institute \\
\footnotesize\texttt{sharang.mudlapura.sriramu@hhi.fraunhofer.de}}
\and
\IEEEauthorblockN{Lampros Gavalakis}
\IEEEauthorblockA{U. of Cambridge \\
\footnotesize\texttt{lg560@cam.ac.uk}}
\and
\IEEEauthorblockN{Ioannis Kontoyiannis}
\IEEEauthorblockA{U. of Cambridge \\
\footnotesize\texttt{ik355@cam.ac.uk}}
\and
\IEEEauthorblockN{Aaron B. Wagner}
\IEEEauthorblockA{Cornell U. \\
\footnotesize\texttt{wagner@cornell.edu}}
}

\maketitle

\begin{abstract}
We establish polarization for memoryless sources taking values on the unit-perimeter circle, or on more general compact groups, under Arıkan’s polar transform. For sources with finite differential entropy, all but an $O((\log N)^{-1/12})$ fraction of the conditional differential entropies in successive cancellation order lie within any fixed positive distance of zero, where $N$ is the block length. Examples indicate that the true polarization rate is indeed
very slow: an inverse-logarithmic rate does not underestimate the polarization speed at moderate block lengths.
\end{abstract}


\section{Introduction}
Given a source string $Y^N=(Y_1,\ldots,Y_N)$ consisting of $n$ independent and identically distributed (i.i.d.) source symbols $Y_i$ taking values in some finite alphabet $\mathcal{Y}$, a \emph{polarizing transform} is a bijection that maps $Y^N$ to a sequence $U^N$ whose conditional entropies, $$H(U_i| U_1,\ldots,U_{i-1}),$$ approach either $0$ or the logarithm of the alphabet size for all but a vanishing fraction of indices as $N$ grows. Arıkan~\cite{arikan2009channel,arikan2010source} identified a rich class of polarizing transforms for Bernoulli sources and introduced a linear construction with a recursive structure that admits a low-complexity implementation. 

Polarization has since been studied for sources over larger discrete alphabets. The linear transform proposed by Arıkan, with addition defined modulo the alphabet size, has been shown~\cite{csacsouglu2009polarization} to be polarizing for prime alphabet sizes, and other polarizing transforms have also been identified~\cite{csacsouglu2012polarization, mori2014source} for non-prime alphabet sizes. Further analysis on the speed of polarization for prime alphabet sizes has shown that the fraction of indices whose conditional entropies remain a fixed distance away from both extreme values decays at least polynomially fast in the block length $N$~\cite{blasiok2022general}. 

In the case of continuous sources, the polarization phenomenon remains poorly understood. Distributions over the real numbers do not lend themselves to an unambiguous definition of polarization analogous to the discrete case, as the conditional differential entropies $h(U_i| U_1,\ldots,U_{i-1})$ are unbounded in both directions; see~\cite{arikan2021entropy} for a discussion on generalizing the notion of polarization to various continuous settings and the challenges involved.

We study polarization for sources taking values in bounded continuous spaces such as the circle $\mathbb{T}=\mathbb{R}/\mathbb{Z}$, which we identify with the interval $[0,1]$; cf.~\cite[Example~5]{arikan2021entropy}. 
This is motivated in part by recent work that constructs lossy compression schemes by combining channel simulation~\cite{li2024channel} with nonlinear transform coding~\cite{balle2020nonlinear}, with the channel simulation being performed using polar codes~\cite{sriramu2024fast}. Nonlinear transform coding and related autoencoding-type architectures are known to struggle with compressing sources with 1-D group structure~\cite{Bhadane:Circle} when the latent variables are real-valued.  One potential remedy is to use latent variables that themselves have a 1-D group structure. Doing so would require a channel simulator, say using polar codes, that simulates channels with this structure (cf.~\cite{sriramu2024fast}). 

Our main theoretical result here (Theorem~\ref{thm:main} in Section~\ref{s:main}) establishes polarization for sources taking values in the circle $\mathbb T$ or in more general compact groups.
Specifically, we show that under the standard polar transform proposed by Arıkan, the fraction of transformed variables whose conditional differential entropy is bounded away from zero vanishes as the block length $N$ grows. Thus, 
most of the indices become nearly uninformative, with most of the information in the source concentrated in the remaining vanishing fraction of variables (cf.~\cite[Example~5]{arikan2021entropy}).

Moreover, we provide an explicit bound on the speed of polarization, showing that the vanishing fraction of informative indices is 
\begin{equation}
O\left(\frac{1}{(\log N)^{1/12}}\right),
\end{equation}
where, throughout, $\log$ denotes the natural logarithm and all familiar information-theoretic
functionals are expressed in nats.
While we certainly do not claim (or believe) that this bound is optimal, empirical evidence discussed in  Sections~\ref{sec:SpeedOfPolarization} and~\ref{speedSec} suggests that the actual rate of polarization is indeed extremely slow, being comparable to $1/(\log N)$.

Our analysis relies on a recent ``entropy-jump'' inequality for random variables on compact groups~\cite{gavalakis2026entropy}, which can be viewed as a stability result for an underlying entropy power inequality (EPI). The EPI-polarization connection is described in the following section; see also~\cite{gavalakis-rupert:26,gavalakis-goh:arxiv} for related recent entropy lower bounds. 

\section{Entropy Inequalities on Compact Groups} \label{sec:SpeedOfPolarization}

For binary random variables, a particularly strong entropy-jump inequality follows from Mrs.\ Gerber's lemma~\cite{WynerZiv1973}.  
From {\c{S}}a{\c{s}}o{\u{g}}lu~\cite[Lemma 2.2]{csacsouglu2012polarization}
and Guruswami and Xia~\cite[(11)-(12)]{GuruswamiXia2015} we have that for i.i.d.\ $X_1,X_2$ on $\mathbb F_2$,
\begin{equation} \label{binaryjump}
 H(X_1+X_2)-H(X_1)
 \geq
 \widehat\theta\,H(X_1)\biggl(1-\frac{H(X_1)}{\log 2}\biggr),
\end{equation}
with
\begin{equation}
 \widehat\theta
 =
 \inf_{0<x<1/2}
 \frac{h(2x(1-x))-h(x)}
      {h(x)(1-h(x))} > 0.799,
\end{equation}
where $h(x):=-x\log x-(1-x)\log(1-x)$ is the binary entropy function.
This quantitative entropy jump is a key ingredient in the latter's analysis of the {speed of polarization}. By a bootstrapping method similar to~\cite{arikan2009rate}, it yields that for $n = O(\log{\frac{1}{\epsilon}})$,
\begin{equation} \label{binaryspeed}
    \mathbb P \bigl( Z_n(W_n^{(I)}) \leq 2^{-\beta n}\bigr) \geq I(X;Y) - \epsilon,
\end{equation}
where $Z_n(W_n^{(I)})$ is the Bhattacharyya parameter of the $I$-th sub-channel, $I(X;Y)$ is the mutual information of the channel, and the probability is over the uniformly selected index $I$.

Interestingly, the later \(q\)-ary analogue of Guruswami and Velingker~\cite{GuruswamiVelingker2015} yields, for prime $q$,
\begin{align}
 &H(X_1+X_2)-H(X_1)
 \nonumber\\&\geq
 \alpha(q)\,H(X_1)\biggl(1-\frac{H(X_1)}{\log q}\biggr),
\end{align}
but the quantitative value of \(\alpha(q)\) furnished by their proof is extremely weak.
In fact, already for \(q=2\), the quantitative estimate resulting from their general argument is vastly weaker than the sharp binary estimate based on Mrs.\ Gerber's lemma. 



In the present work we use the entropy-jump inequality of~\cite[Theorem 2.10]{gavalakis2026entropy}. That estimate is significantly weaker but much more general than the entropy jump of the binary case~\eqref{binaryjump}. In addition, the related entropy-jump obtained in~\cite[Theorem 1.4]{gavalakis2026entropy} involves an extremely weak constant, with a value no greater than $3\times 10^{-5}$. Although we do not know whether sharper estimates exist, it is worth noting that the actual speed of polarization in the continuous setting seems to be quite slow as we discuss in Section~\ref{speedSec}. 

\section{Main Result} 
\label{s:main}

Let $\mathcal{X}$ be a standard Borel space, and let $\mathcal{Y}$ be
a connected, compact, metrizable abelian topological group, equipped with its Borel $\sigma$-algebra. Let $\mu$ denote the normalized Haar measure on $\mathcal{Y}$, so that $\mu(\mathcal{Y})=1$. All familiar information-theoretic functionals are defined in terms of $\mu$. For example, suppose $X$ takes values in ${\cal Y}$ and has law $P$. If $P$ has a density $f=\frac{dP}{d\mu}$ with respect to $\mu$, then the (differential) entropy is defined as $h(X)=-\int f\log f\,d\mu$; otherwise, we set $h(X)=-\infty$. In particular,  the uniform distribution on $\mathcal{Y}$ has zero entropy and all
entropies are nonpositive.

Let $P_{XY}$ be a joint distribution on $\mathcal{X} \times \mathcal{Y}$. We will refer to the corresponding marginal $P_X$ as the \emph{source distribution} and the induced kernel $P_{Y|X}$ as the \emph{channel}. For any such joint distribution, we will assume throughout that the mutual information $I(X;Y)$ associated with the channel is finite and in fact that 
\begin{equation} -\infty<h(Y| X) \leq 0.\end{equation}

Given $N=2^n$ i.i.d. pairs $(X^N, Y^N) \sim P_{XY}^{\otimes N}$, we define the polar domain representation~\cite{arikan2010source} as
\begin{equation}
    U^N = Y^NB_N{\begin{bmatrix}
        1 & 0\\
        1 & 1
    \end{bmatrix}}^{\otimes n }, \label{eq:HadamardPT}
\end{equation} 
where $B_N$ is the bit-reversal permutation.
Our main contribution is the following finite-$n$ result.

\begin{theorem} \label{thm:main}
    For all $n>1$ and $0 < \epsilon < \min\left\{ -h(Y| X), \frac{1}{8} \right\}$, we have
    \begin{align}
        \frac{\left| \{i: - h(U_i| U_1^{i-1}, X^N) \ge \epsilon \} \right|}{2^n} \le \frac{4 (-h(Y| X))^{1/6}}{(n\epsilon^{10})^{1/12}}, \label{eq:main}
    \end{align}
    where $N = 2^n$.
\end{theorem}

\section{Proof}
We begin by examining a single recursive step of the polar transform.

\begin{definition}[Pairwise polarizing transform]
Let $Z_1$ and $Z_2$ be independent random variables with distribution $P_Z$ on some standard Borel space $\mathcal{Z}$. Given an arbitrary channel $W: \mathcal{Z} \to \mathcal{Y}$ and conditionally independent channel outputs $Y_1 \sim W(\cdot| Z_1)$ and $Y_2 \sim W(\cdot| Z_2)$, define their polar domain representations
\begin{align}
    U_1 &= Y_1 + Y_2, \label{eq:PT1}\\
    U_2 &= Y_2. \label{eq:PT2}
\end{align}
For an arbitrary measurable subet $A$ of $\mathcal{Y}$, define the transformed channels
\begin{multline}
    W_0(U_1 \in A| z_1, z_2) \\ = \Pr\left(Y_1 + Y_2 \in A| Z_1 = z_1, Z_2 = z_2 \right),  
    \end{multline}
    and
    \begin{multline}
    W_1(U_2 \in A| z_1, z_2, u_1) \\ = \Pr\left(Y_2 \in A| Z_1 = z_1, Z_2 = z_2, U_1 = u_1 \right).
\end{multline}
We will denote this transform by
\begin{align}
    (W_0, W_1) \leftarrow W.
\end{align}
For a channel $\hat{W}$ together with its induced input-output
pair $(\hat{Z},\hat{Y})$, we use the abbreviation
\begin{equation}
    h(\hat{W}) = h(\hat{Y}|\hat{Z}).
\end{equation}
In particular,
\begin{align}
    h(W)
    &=h(Y_1| Z_1)
      =h(Y_2| Z_2),                            \\
    h(W_0)
    &=h(U_1| Z_1,Z_2)
      =h(Y_1+Y_2| Z_1,Z_2),                    \\
    h(W_1)
    &=h(U_2| Z_1,Z_2,U_1)
      =h(Y_2| Z_1,Z_2,Y_1+Y_2).
\end{align}
We further write
\begin{equation}
    h(\hat{W}| z) = h(\hat{Y}|\hat{Z}=z)
\end{equation}
for the pointwise output entropy corresponding to the input
$\hat{Z}=z$.
\end{definition}

Starting from the channel $P_{Y|X}$, we define a binary tree in the usual fashion:
\begin{align}
    W &= P_{Y|X},\\
    (W_0, W_1) &\leftarrow W,\\
    (W_{b_1,\ldots,b_K,0}, W_{b_1,\ldots,b_K,1}) &\leftarrow W_{b_1,\ldots,b_K},
\end{align}
where $K \ge 1$ and each $b_i \in \{0,1\}$. We then draw i.i.d. uniform Bernoulli samples $B_1,B_2,\ldots$ to construct a random walk on the tree, with respect to which we can define the stochastic process
\begin{align}
    I_0 &= -h(W),\\
    I_K &= -h(W_{B^K}), \, K \ge 1.
\end{align}
We can express our objective in terms of $I_k$:
\begin{align}
    \Pr(I_n \ge \epsilon) &= \frac{\left| \{i: - h(U_i| U_1^{i-1}, X^N) \ge \epsilon \} \right|}{2^n}. \label{eq:Target}
\end{align}

\begin{lemma} \label{lemma:IConv}
    The process $(I_k)_{k\ge0}$ is a martingale that converges almost surely.
\end{lemma}
\begin{proof}
    The martingale property of $I_k$ follows from the fact that every pairwise polarizing transform is a entropy-preserving transform. Since the uniform distribution over $\mathcal{Y}$ is the maximum entropy distribution, it follows that $I_k \ge 0$. Convergence then follows from Doob's supermartingale convergence theorem.
\end{proof}
We consider the case where 
\begin{equation}
    n \epsilon^{10} > (4I_0^{1/6})^{12}, \label{eq:n_epsilon_condition}
\end{equation}
since the bound in~(\ref{eq:main}) holds trivially otherwise.
We begin by defining two thresholds $M$ and $\tilde{\epsilon}$,
\begin{align}
    M &= \left(\frac{10^{10}I_0^{10}n \epsilon^{10}}{2312}\right)^{1/12},  \label{eq:Mdef}\\
    \tilde \epsilon &= \left(\frac{231200 I_0^2 \epsilon^2}{n}\right)^{1/12}. \label{eq:epsdef}
\end{align}
The condition from~(\ref{eq:n_epsilon_condition}) ensures that
\begin{align}
    M > I_0 \text{ and } \tilde \epsilon < \epsilon.
\end{align}
We also define the stopping times
\begin{align}
    \tau &= \inf\{K: I_K \le \tilde \epsilon \text{ or } I_K>M\},\\
    \sigma &= \inf\{K: I_K>M\},\\
    \omega &= \inf \left\{K: I_K \le \tilde \epsilon\right\}.
\end{align}
and the stopped martingales $J_K = I_{K \wedge \tau}$ and $\theta_K = I_{K \wedge \sigma}$. We then have
\begin{align}
    &\Pr(I_n \ge \epsilon)\nonumber\\
    &\le \Pr(I_n \ge \epsilon, \sigma > n) + \Pr(\sigma \le n)\nonumber\\
    &\le \Pr(I_n \ge \epsilon, \sigma > n, \tau < n) + \Pr(\tau \ge n) + \Pr(\sigma \le n)\nonumber\\
    &\le \Pr(I_n \mathbb{I}\{\omega < n\} \ge \epsilon) +  \Pr(\tau \ge n) + \Pr(\sigma \le n). \label{eq:AchProofMain}
\end{align}
Next, we derive upper bounds for each of the quantities in~(\ref{eq:AchProofMain}).

\paragraph{Upper bound on $\Pr(\tau \ge n)$} \label{par:tau}
For all $K > 1$, using the orthogonality of martingale increments, 
\begin{align}
    E\left[J_{K+1}^2\right] &\ge E\left[\sum\limits_{i=0}^K (J_{i+1} - J_i)^2\right]\nonumber\\
    &= E\left[\sum\limits_{i = 0}^{K \wedge (\tau-1)} (I_{i+1} - I_i)^2\right].
\end{align}
The lower bound from Lemma~\ref{lemma:StepwiseBound} below gives
\begin{align}
    E[K \wedge \tau] \le \frac{1156 E\left[I_{(K+1) \wedge \tau}^2\right]}{{\tilde \epsilon}^{10}}.
\end{align}
From the definition of the stopping time, it follows that $I_{(K+1) \wedge \tau} \le 2M$. Therefore,
\begin{equation}
    E\left[I_{(K+1) \wedge \tau}^2\right] \le 2ME\left[ |I_{(K+1) \wedge \tau}| \right] = 2MI_0,
\end{equation}
and hence,
\begin{equation}
    E[K \wedge \tau] \le \frac{2312MI_0}{{\tilde \epsilon}^{10}}. \label{eq:Boundtau}
\end{equation}
Now, observe that
\begin{equation}
    \Pr(\tau \ge K) = \Pr\left(\tau \wedge K = K\right)
    \le \Pr\left(\tau \wedge K \ge K\right).
\end{equation}
Applying Markov's inequality with $K=n$, we obtain
\begin{equation}
    \Pr(\tau \ge n) \le \frac{E\left[ \tau \wedge n\right]}{n}
    \le \frac{2312MI_0}{n{\tilde \epsilon}^{10}},
\end{equation}
where the second step follows from~\eqref{eq:Boundtau}.

\paragraph{Upper bound on $\Pr(\sigma \le n)$} \label{par:sigma}

For all $K$, we have
\begin{align}
    \theta_K &\ge M \mathbb{I}\left\{\sigma \le K\right\}.
\end{align}
Taking expectations and rearranging, we obtain
\begin{equation}
    \Pr\left(\sigma \le K\right) \le \frac{E[\theta_K]}{M} = \frac{I_0}{M},
\end{equation}
for all $K$, and hence,
\begin{equation}
        \Pr\left(\sigma \le n\right) \le \frac{I_0}{M}. \label{eq:BoundSigma}
\end{equation}

\paragraph{Upper bound on $\Pr(I_n \mathbb{I}\{\omega < n\} \ge \epsilon)$}
Markov's inequality gives
\begin{align}
    \Pr(I_n \mathbb{I}\{\omega < n\} \ge \epsilon) &\le \frac{E\left[I_n \mathbb{I}\{\omega < n]\right\}}{\epsilon}. \label{eq:LowerExitBound}
\end{align}
Since we have that
$E[I_K - I_{K \wedge \omega}] = 0$ for all $K$,
it follows that $E\left[I_n - I_{n \wedge \omega}\right] = 0$, which in turn implies
that $E\left[(I_n - I_{n \wedge \omega})\mathbb{I}\{\omega < n\}\right] = 0$,
and therefore,
\begin{equation}
    E\left[I_n \mathbb{I}\{\omega < n\}\right] = E\left[I_{n \wedge \omega} \mathbb{I}\{\omega < n\}\right] \le \tilde \epsilon.
\end{equation}
Substituting this in~(\ref{eq:LowerExitBound}), we obtain
\begin{align}
    \Pr(I_n \mathbb{I}\{\omega < n\} \ge \epsilon) \le \frac{\tilde \epsilon}{\epsilon}.
\end{align}

Using the above upper bounds in~(\ref{eq:AchProofMain}), gives
\begin{align}
    \Pr(I_n \ge \epsilon) &\le \frac{2312MI_0}{n{\tilde \epsilon}^{10}} + \frac{I_0}{M} + \frac{\tilde \epsilon}{\epsilon}.
\end{align}
And substituting the definitions from~(\ref{eq:Mdef}) and~(\ref{eq:epsdef}) yields
\begin{align}
    \Pr(I_n \ge \epsilon) &\le 4I_0^{1/6} \cdot(n \epsilon^{10})^{-1/12},
\end{align}
completing the proof of the theorem.
\hfill$\Box$

\begin{remark}{\em
The results of Lemma~\ref{lemma:IConv} and Theorem~\ref{thm:main} imply that the martingale $(I_k)_{k\ge0}$ converges to $0$ almost surely:
    \begin{equation}
        I_\infty := \lim\limits_{n \to \infty} I_n = 0 \text{ a.s.}
    \end{equation}
    On the other hand, whenever $I_0>0$,
    \begin{equation}
        \lim\limits_{n \to \infty} E[|I_n - I_\infty|] = \lim\limits_{n \to \infty} E[I_n]
        = I_0,
    \end{equation}
therefore $(I_k)_{k\ge0}$ is not uniformly integrable and does not converge in $L_1$.
See the discussion and conjecture in~\cite[Example~5]{arikan2021entropy}.
}\end{remark}

\begin{figure*}[t]
    \centering
 \scalebox{0.7}{\includegraphics{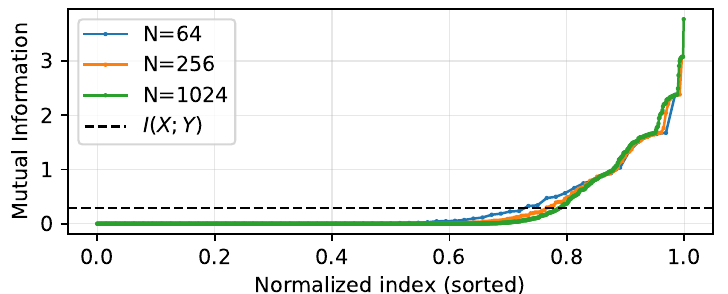}}
        \caption{Mutual information per subchannel, as
        defined in~(\ref{eq:MI}), for the example in Sec.~\ref{speedSec}. For any $\epsilon > 0$, the fraction of subchannels for which the mutual
        information exceeds
        $\epsilon$ tends to zero slowly with $N$. 
        The dotted line shows the mutual information
        of the unpolarized channel in~(\ref{eq:origchannel}).
        Note that
        the mutual information of the best subchannel increases without bound.}
        \label{fig:first}
    \end{figure*}
    
    \begin{figure}[t]
        \centering
        \scalebox{.7}{\includegraphics{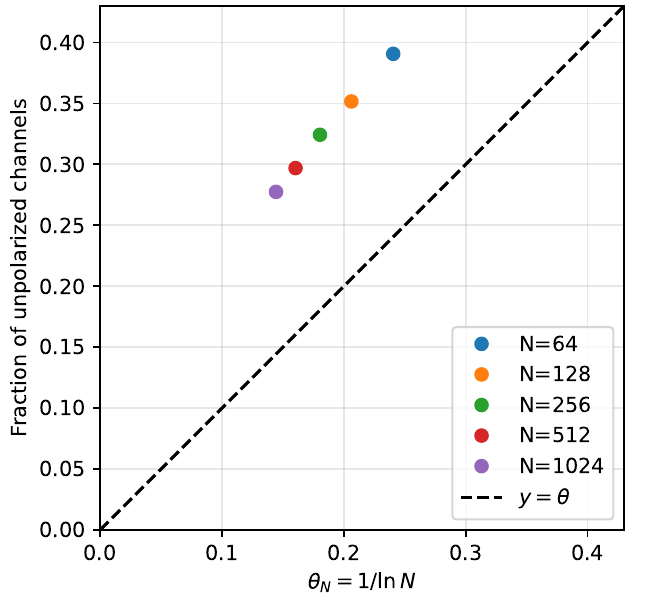}}
        \caption{The fraction of subchannels whose mutual
        information exceeds $\epsilon = 0.05$ nats plotted
        against the function $1/\ln N$. }
        \label{fig:second}
    \end{figure}

\begin{lemma} \label{lemma:StepwiseBound}
    Given $0 < \epsilon < \frac{1}{8}$, let $b^K \in \{0,1\}^K$ be a binary string such that
    \begin{equation}
        -h(W_{b^K}) \ge \epsilon.
    \end{equation}
    Then, the transformed channels
    \begin{align}
        (W_{b^K,0}, W_{b^K,1}) \leftarrow W_{b^K},
    \end{align}
    satisfy
    \begin{align}
        h(W_{b^K,0}) - h(W_{b^K}) = h(W_{b^K}) - h(W_{b^K,1}) \ge \frac{\epsilon^5}{34}.
    \end{align}
\end{lemma}
\begin{proof}
    For almost surely-any pair of inputs $z_1$ and $z_2$ to the channels $W_{b^K}$, applying the entropy-jump inequality in~\cite[Theorem~2.10]{gavalakis2026entropy} yields
    \begin{align}
        &h(W_{b^K,0}| z_1, z_2) - \frac{h(W_{b^K}| z_2) + h(W_{b^K}| z_1)}{2} \nonumber\\
        &\qquad\ge \max\left\{ f(h(W_{b^K}| z_1)), f(h(W_{b^K}| z_2))\right\},
    \end{align}
    where
    \begin{align}
        f(x) = \begin{cases}
            \frac{-x^5 - 4x^6}{16 + 8x^4}, &-\frac{1}{8} \le x \le 0,\\
            f\left( \frac{-1}{8} \right), &x<-\frac{1}{8}.
        \end{cases}
    \end{align}
    Averaging over independent realizations of $z_1$ and $z_2$ and observing that, for $-\frac{1}{8} \le x \le 0$,
    \begin{align}
        \frac{-x^5 - 4x^6}{16 + 8x^4} \ge -\frac{x^5(1 + 4x)}{16 + \frac{1}{2}} \ge -\frac{x^5}{33},
    \end{align}
    we obtain, using Jensen's inequality,
    \begin{align}
        &h(W_{b^K,0}) - h(W_{b^K})\nonumber \\
        &\quad \ge E_{Z_1}\left[ \min \left\{ -\frac{h(W_{b^K}| Z_1)^5}{33}, \frac{(1/8)^5}{33} \right\}\right]
        \nonumber\\
        &\quad\geq \frac{\left(E_{Z_1}\left[\min\left\{ -h(W_{b^K}| Z_1), \frac{1}{8}\right\}\right]\right)^5}{33}. \label{eq:BoundWithMin}
    \end{align}
    Now, we can further bound,
    \begin{align}
        &E_{Z_1}\left[\min\left\{ -h(W_{b^K}| Z_1), \frac{1}{8}\right\}\right] \nonumber\\
        &= E_{Z_1}\left[\max\left\{ -h(W_{b^K}| Z_1), \frac{1}{8}\right\} - \left| h(W_{b^K}| Z_1) + \frac{1}{8}\right|\right]\nonumber\\
        &\ge \max\left\{ -h(W_{b^K}), \frac{1}{8}\right\} \nonumber\\
        &\phantom{=}- E_{Z_1}\left[\left| h(W_{b^K}| Z_1) + \frac{1}{8} - h(W_{b^K}) + h(W_{b^K})\right|\right]\nonumber\\
        &\ge \max\left\{ -h(W_{b^K}), \frac{1}{8}\right\} - \left|\frac{1}{8} + h(W_{b^K})\right| \nonumber\\
        &\phantom{=}- E_{Z_1}\left[\left| h(W_{b^K}| Z_1) - h(W_{b^K}))\right|\right]\nonumber\\
        &= \min\left\{ -h(W_{b^K}), \frac{1}{8}\right\} - E_{Z_1}\left[\left| h(W_{b^K}| Z_1) - h(W_{b^K}))\right|\right]\nonumber\\
        &\ge \epsilon - E_{Z_1}\left[ \left| h(W_{b^K}| Z_1) - h(W_{b^K})\right| \right]\nonumber\\
        &\ge \epsilon - E_{Z_1, Z_2}\left[ \left| h(W_{b^K}| Z_1) - h(W_{b^K}| Z_2)\right| \right]\nonumber\\
        &= \epsilon - 2E_{Z_1, Z_2}\Big[ \max\left\{ h(W_{b^K}| Z_1) , h(W_{b^K}| Z_2)\right\} \nonumber\\
        &\phantom{=} \hspace{1.1in} - \frac{h(W_{b^K}| Z_1) + h(W_{b^K}| Z_2)}{2} \Big].
        \label{eq:MinBoundL4}
        \end{align}
        From Jensen's inequality we have that the entropy
        jump is nonnegative:
        \begin{equation}
        \max\{ h(W_{b^K}| Z_1) , h(W_{b^K}| Z_2)\} \le
        h(W_{b^K,0}|Z_1,Z_2).
        \end{equation}
        Using this and the identity 
        \begin{equation}
        E[h(W_{b^K}| Z_1)] =  E[h(W_{b^K}| Z_2)],
        \end{equation}
        gives
        \begin{multline}
        E_{Z_1}\left[\min\left\{ -h(W_{b^K}| Z_1), \frac{1}{8}\right\}\right] \\
        \ge \epsilon - 2\left( h(W_{b^K,0}) - h(W_{b^K}) \right). \label{eq:AbsUB}
    \end{multline}
    Substituting~(\ref{eq:AbsUB}) in~(\ref{eq:BoundWithMin}), we obtain
    \begin{align}
        &h(W_{b^K,0}) - h(W_{b^K})\nonumber\\
        &\qquad \ge \frac{\left(\epsilon - 2\left( h(W_{b^K,0}) - h(W_{b^K}) \right)\right)_+^5}{33}.
    \end{align}
    The function $t\mapsto(t_+)^5$ is convex on the whole real line. Its supporting-line inequality at $t=\epsilon$ therefore gives
\begin{align}
&(\epsilon-2( h(W_{b^K,0}) - h(W_{b^K})))_+^5
    \nonumber\\
&\quad \ge\epsilon^5-10\epsilon^4\left( h(W_{b^K,0}) - h(W_{b^K}) \right).
\end{align}
Rearranging terms,
    \begin{align}
        h(W_{b^K,0}) - h(W_{b^K}) &\ge \frac{\epsilon^5}{33 + 10 \epsilon^4} \ge \frac{\epsilon^5}{34},
    \end{align}
    which is the required result.
\end{proof}

\section{Examples} \label{speedSec}

Given a uniform source $X \sim \mathrm{Unif}[0,1]$ on the circle~$\mathbb{T}$ identified as the interval $[0,1]$ with modulo-1 addition,  consider the additive cosine channel
\begin{equation}
\label{eq:origchannel}
    Y = X + Z \mod 1,
\end{equation}
where the noise $Z$ is independent of $X$ and has density 
\begin{equation}
    p_Z(\theta) = 1 + \alpha \cos(2\pi \theta), 
\end{equation}
for some $|\alpha|\leq 1$. Under the polar transform in (\ref{eq:HadamardPT}) applied to $N$ independent copies of $X$ and $Y$, we compute the subchannel densities using their Fourier parameterization~\cite{de2024fourier} as
\begin{multline}
    P(U_i|U_{1}^{i-1}, X^N) \\= 1 + 2\Re \left(\sum\limits_{k=1}^Nc_k(X^N, U_1^{i-1})\exp\left(j2\pi kU_i\right)\right).
\end{multline}
This representation includes the original channel density and is preserved under every recursive update in~\eqref{eq:PT1}-\eqref{eq:PT2}.

For several values of $N$ and $\alpha=0.98$, Fig.~\ref{fig:first} shows Monte Carlo estimates of the \emph{subchannel mutual informations}
\begin{equation}
\mathrm{MI}_i = 
 h(U_i) -h(U_i |U_1^{i-1}, X^N) =
 -h(U_i | U_1^{i-1}, X^N).
 \label{eq:MI}
\end{equation}
Each estimate is averaged over $100$ independent trials, and the $N$ subchannels are ordered by their estimated mutual information. Fig.~\ref{fig:second} plots the fraction of the unpolarized channels, defined here as the fraction of indices $i$ for which $\mathrm{MI}_i > 0.05$ nats, for increasing values of $N$. The upper bound on the polarization rate from Theorem~\ref{thm:main} is vacuous for the small values of $N$ considered here and is thus omitted from the plot. The reference curve in the figure suggests that the actual rate of polarization in this example is closer to $\frac{1}{\log N}$.

It is easy to see that the rate of polarization can be no faster than $1/\log N$ in the worst case. Suppose $X = 0$ and $Y$ is uniformly distributed over $[0,1/2]$.  Then for $i<N$, the component with the smallest conditional entropy, $U_N$, has a conditional distribution given $U_i$ that is uniformly distributed over an interval of approximate width $1/N$ with high probability. Thus we have
\begin{equation}
- h(U_N|U_1,\ldots,U_{N-1}) = \Theta(\log N),
\end{equation}
and $-h(U_i|U_1,\ldots,U_{i-1}) = O(\log N)$
for all $i$. Since the average conditional entropy
satisfies
\begin{equation}
\frac{1}{N}
\sum_{i = 1}^N h(U_i|U_1,\ldots,U_{i-1}) = - \log 2
\end{equation}
it follows a fraction $\Omega(1/\log N)$ of the subchannels must have entropy bounded away from zero. 

Thus one cannot hope to achieve fast polarization in general
using Arıkan’s polar transform. Whether the polarization
can be accelerated using other transforms, at least
in examples of interest in applications, is a subject
of ongoing work.

\newpage

\section*{AI Statement}
OpenAI's ChatGPT 6.0~\cite{openai2026chatgpt} was used to assist with editing the manuscript, exploring proof techniques, and developing the code for the simulations in Section~\ref{speedSec}. In particular, the argument for the lower bound in equations~(\ref{eq:MinBoundL4})--(\ref{eq:AbsUB}) was suggested by ChatGPT. The authors take full responsibility for the contents in this manuscript.

\IEEEtriggeratref{9}

\bibliographystyle{plain}
\bibliography{ref.bib}

\end{document}